%% file: main.tex
\documentclass[a4paper,UKenglish,cleveref,thm-restate,numberwithinsect]{article}
\input{makro}

\nolinenumbers
\newcommand{\bvec}[1]{\underline{\mathbf{#1}}}
\newcommand{\bx}{\bvec{x}}
\newcommand{\by}{\bvec{y}}
\newcommand{\bz}{\bvec{z}}
\newcommand{\bu}{\bvec{u}}

\newcommand{\boldmu}{\boldsymbol{\mu}}

\title{Subexponential Algorithms for Total Search Problems in the Dense Pigeonhole Regime}
\author{Jadwiga Czyżewska \thanks{Institute of Informatics, University of Warsaw,
Warsaw, Poland. 
Supported by Polish National Science Centre SONATA BIS-12 grant
No.~2022/46/E/ST6/00143.}
\and  Karolina Drabik \thanks{Institute of Informatics, University of Warsaw,
Warsaw, Poland 
Supported by the European Research Council (ERC), BUKA,
grant agreement No.~101126229.}
\and Marek Sokołowski \thanks{ETH Z\"urich, Z\"urich, Switzerland}
\and Karol Węgrzycki \thanks{Max Planck Institute for Informatics,
Saarland Informatics Campus, Saarbr\"ucken, Germany
Supported by the Deutsche Forschungsgemeinschaft
(DFG, German Research Foundation), grant No.~559177164.}
}

\begin{document}
\maketitle
\thispagestyle{empty}

\begin{abstract}
The class $\mathsf{PPP}$ captures fundamental total search problems, for which
a solution is guaranteed to exist by the pigeonhole principle and the main challenge is to find it.
We ask whether problems in this class admit significantly faster algorithms in the
\emph{dense regime}.

More concretly, we study \emph{binary collision Short Integer Solution} (SIS): given a matrix
$A\in\mathbb{Z}_q^{n\times m}$ with $m>n\log q$,
the goal is to find distinct $\bx,\by\in\{0,1\}^m$ with $A\bx=A\by\pmod q$. Brute-force enumeration can be used to find a collision in $q^{\mathcal{O}(n)} \operatorname{poly}(|I|)$ time, and no significantly faster algorithm for that problem is known.
We present an algorithm running in time
$q^{o(n)}\cdot\operatorname{poly}(|I|)$ whenever the density parameter $\delta\coloneqq m/(n\log q)$ is $\omega(1)$.\footnote{All logarithms are in base $2$, and
$|I|$ is the binary input length.}
Building on reductions of Sotiraki et al.~(FOCS~2018),
we obtain corresponding algorithms for dense variants of weak constrained
SIS, Equal Vector Sums, Minkowski, and Blichfeldt.

Our algorithm adapts techniques of Bhati et al.~(FOCS~2026), who
gave a subexponential algorithm for Equal Sums in the dense regime and proved a tight conditional lower bound.
Their conditional lower bound easily transfers to SIS and shows that the running time of our algorithms is essentially optimal.
\end{abstract}

\clearpage
\setcounter{page}{1}
\hypersetup{pageanchor=true}
\input{intro}
\input{prelim}
\input{main_proof}
\input{corollary_proofs}

\section*{Use of generative AI}

We used OpenAI Codex and GPT-6 to proofread and lightly edit the manuscript (for example, to
draw \cref{fig:pigeonhole-classes}). All the ideas and proofs in this paper were
developed by humans on whiteboards.

\bibliographystyle{plainurl}
\bibliography{bibliography}
\end{document}

%% file: makro.tex
\usepackage[margin=1.25in]{geometry}
\usepackage{libertine}
\usepackage[T1]{fontenc}
\usepackage[utf8]{inputenc}
\usepackage{amsfonts,amssymb,amsmath,amsthm}
\usepackage{calc}
\usepackage{tikz}
\usetikzlibrary{calc}
\usepackage[bottom,para,symbol]{footmisc}
\usepackage{hyperref}
\usepackage{thmtools}
\usepackage{thm-restate}

\usepackage{pdfpages}
\usepackage{epigraph}
\usepackage{lineno}

\usepackage{url}
\usepackage{hyperref}
\usepackage{enumitem}
\usepackage{xspace}
\usepackage{todonotes}
\usepackage{thmtools}
\usepackage{mathrsfs}  
\usepackage{textpos}
\usepackage[normalem]{ulem}
\usepackage[capitalize, nameinlink]{cleveref}

\newtheorem{theorem}{Theorem}[section]

\newtheorem{claim}[theorem]{Claim}

\newtheorem{observation}[theorem]{Observation}

\theoremstyle{definition}
\newtheorem{definition}[theorem]{Definition}

\usepackage[sort,nocompress]{cite}
\usepackage{mathtools}
\usepackage{xspace}
\usepackage{tikz}
\usepackage{hyperref}
\usepackage{amsthm}
\usetikzlibrary{calc}
\usepackage[framemethod=TikZ]{mdframed}
\usepackage{todonotes}
\usepackage[linesnumbered,ruled,vlined]{algorithm2e}
\DontPrintSemicolon

\newenvironment{problemfields}{%
\begin{list}{}{%
\setlength{\labelwidth}{6em}%
\setlength{\labelsep}{0.6em}%
\setlength{\leftmargin}{6.6em}%
\setlength{\rightmargin}{0pt}%
\setlength{\itemindent}{0pt}%
\setlength{\listparindent}{0pt}%
\setlength{\topsep}{0pt}%
\setlength{\partopsep}{0pt}%
\setlength{\parsep}{0pt}%
\setlength{\itemsep}{0.25\baselineskip}%
}}{\end{list}}
\newcounter{prob}
\newenvironment{prob}[2][]{%
\refstepcounter{prob}%
\mdfsetup{frametitle={%
\tikz[baseline=(current bounding box.east),outer sep=0pt]
\node[anchor=east,rectangle,fill=white]
{\strut Problem~\arabic{prob}:~#1};},%
innertopmargin=-0.5em,linecolor=black,linewidth=1pt,roundcorner=6pt,topline=true,%
frametitleaboveskip=\dimexpr-\ht\strutbox\relax}%
\begin{mdframed}[nobreak=true]\relax%
\label{#2}\begin{problemfields}}{\end{problemfields}\end{mdframed}}
\crefname{prob}{problem}{problems}
\Crefname{prob}{Problem}{Problems}

\newcommand{\Oh}{\mathcal{O}}

\newcommand{\xor}{\oplus}

\newcommand{\floor}[1]{\left\lfloor #1\right\rfloor}
\newcommand{\ceil}[1]{\left\lceil #1\right\rceil}

\newcommand{\set}[1]{\{#1\}}


%% file: intro.tex
\section{Introduction}

The class $\mathsf{PPP}$, introduced by Papadimitriou~\cite{Papadimitriou},
formalizes search problems based on the pigeonhole principle.  Its defining
problem asks for either a \emph{collision} (i.e., two distinct inputs with the
same output) or an input that maps to zero, in a function whose input and output
spaces have the same size. The pigeonhole principle guarantees that one of
these outcomes exists, but finding it can still be hard.

The class $\mathsf{PPP}$ contains many natural problems, including Equal Sums,
variants of Short Integer Solution (SIS), Minkowski, and
Blichfeldt~\cite{Papadimitriou,DBLP:conf/focs/SotirakiZZ18}.\footnote{In the Equal Sums problem, we are given $m$ positive integers with a total sum less than $2^m-1$, and the goal is to find two distinct
subsets with the same sum; see \Cref{sec:prelim}.}
It has a natural
connection to cryptography: for example, collision-resistant hash functions fit
naturally into this class~\cite{DBLP:conf/focs/SotirakiZZ18}.  We define the full
versions of these problems in \Cref{sec:our-results} and the formal class
definitions and lattice background are in \Cref{sec:prelim}.

Each of these problems has variants in the subclass $\mathsf{PWPP}$, including
the versions studied below.  This class uses the \emph{weak pigeonhole
principle}: the input circuit has fewer output bits than input bits, which
guarantees a collision. This class has a $\mathsf{PWPP}$-complete problem:
Sotiraki, Zampetakis, and Zirdelis~\cite{DBLP:conf/focs/SotirakiZZ18} introduced
\emph{weak constrained SIS} and proved that every problem in $\mathsf{PWPP}$ reduces to
it in polynomial time.

Generic algorithms for these classes enumerate or randomly sample inputs,
with exponential worst-case running time in the number of output bits.
We therefore ask if we can have better algorithms when the number of
input bits is much larger than the number of output bits.
This is the \emph{dense pigeonhole regime}:
\begin{center}
\emph{Can all these natural PPP problems be solved faster in the dense pigeonhole regime?}
\end{center}

Recently, Bhati et al.~\cite{BhatiDJRW26} gave an algorithm for dense Equal Sums
instances. For $m$ positive integers bounded by $M$, they obtain
$M^{o(1)}\cdot\operatorname{poly}(m)$ time when $m=\omega(\log M)$.  Earlier
algorithms used the pigeonhole promise to improve the dependence on
$m$~\cite{JW24,JWZ25}. We extend their method to binary collision SIS
and obtain an algorithm for weak constrained SIS, whose general form is
$\mathsf{PWPP}$-complete~\cite{DBLP:conf/focs/SotirakiZZ18}.
Explicit reductions give algorithms for dense variants of Equal Vector Sums,
Minkowski, and Blichfeldt as well.

\input{figure1}
\subsection{Our results}\label{sec:our-results}

Throughout, $n$ is positive, $q\geq2$ is a power of two, and logarithms
have base $2$. We write $|I|$ for the binary input length.
We first study the following collision variant of SIS.

\begin{prob}[Binary collision SIS]{def:sis}
\item[Input:] A matrix $A\in\mathbb{Z}_q^{n\times m}$
with $m>n\log q$.
\item[Output:] Two distinct binary vectors
$\bx,\by\in\{0,1\}^m$ such that
$A\bx=A\by\pmod q$.
\item[Parameters:] $n$, $q$, and $\delta\coloneqq m/(n\log q)$.
\end{prob}

There are $2^m$ binary vectors and only $q^n$ possible images under $A$.
Since $m>n\log q$, the pigeonhole principle guarantees a collision.
In cryptographic SIS, the matrix is usually sampled at random and
larger coordinates may be allowed. Our algorithm works for every input
matrix.

Evaluating $A$ on $\Oh(q^n)$ distinct binary vectors and comparing their
images finds a collision in time
$q^n\cdot\operatorname{poly}(|I|)$. Our main result improves this
exponential factor in the dense regime.

\begin{restatable}{theorem}{SISalg}\label{thm:SISalg}
An instance $I$ of binary collision SIS (\Cref{def:sis}) with
parameters $n$, $q$, and $\delta$
can be solved deterministically in time
$q^{\mathcal{O}(n/\log\delta)}\cdot\operatorname{poly}(|I|)$.
\end{restatable}

When $\delta=\omega(1)$, the running time is
$q^{o(n)}\cdot\operatorname{poly}(|I|)$, with subexponential dependence on
$n\log q$.  In the other regime, our algorithm is even polynomial 
whenever $n=\mathcal\Oh((\log m)^2/\log q)$. The same polynomial-time guarantee holds
for all the consequences below (using their respective density parameters). To
the best of our knowledge, this guarantee was not previously known for any of
the variants presented below.

We would like to note that improving the dependence on $\log\delta$ further is
unlikely with current techniques: as observed by Bhati et al.~\cite{BhatiDJRW26},
any improvements would yield a breakthrough for the Successive Minima Problem (SMP).

\begin{restatable}[cf.~\cite{BhatiDJRW26}]{lemma}{SISlowerbound}\label{lem:sis-lower-bound}
Any algorithm that solves binary collision SIS
in time $q^{o(n/\log\delta)}\cdot\operatorname{poly}(|I|)$
would improve the best known algorithms for SMP, even when $\delta = \operatorname{poly}(|I|)$ and $n=\Oh(1)$.
\end{restatable}
We recall the SMP and the proof in \Cref{sec:lattice-lower-bound}.

\paragraph*{Consequences}\label{sec:applications}

Weak constrained SIS, introduced
in~\cite{DBLP:conf/focs/SotirakiZZ18}, adds linear constraints to binary
collision SIS: both output vectors must also satisfy $G\bx=G\by=0$.
The matrix $G$ is \emph{binary invertible}, which allows us to choose
some coordinates of the vectors freely and then compute the remaining coordinates
uniquely (see~\Cref{def:binary-invertible} for a formal definition).

\begin{prob}[Weak constrained SIS]{def:wcsis}
\item[Input:] Matrices $A\in\mathbb{Z}_q^{n\times m}$ and
$G\in\mathbb{Z}_q^{d\times m}$, where $G$ is binary invertible
and $m>(n+d)\log q$.
\item[Output:] Two distinct binary vectors
$\bx,\by\in\{0,1\}^m$ such that
$A\bx=A\by$ and
$G\bx=G\by=0\pmod q$.
\item[Parameters:] $n$, $m$, $d$, $q$, and $\delta\coloneqq (m-d\log q)/((n+d)\log q)$.
\end{prob}

Binary invertibility of $G$ lets us construct valid output vectors $\bx$ in the following way:
choose $m-d\log q$ coordinates of $\bx$ freely.
Then, by the design of $G$, each choice uniquely determines the other $d\log q$ coordinates.
Thus, exactly $2^{m-d\log q}>q^n$ binary vectors satisfy $G\bx=0$, and
two of them must have the same image under $A$.
The same argument works with $0$ replaced by any target
$\mathbf{b}\in\mathbb Z_q^d$.

\begin{restatable}{corollary}{denseWcSIS}\label{cor:dense-wcsis}
If $\delta>1$, weak constrained SIS (Problem~\ref{def:wcsis}) can be solved in time
$q^{\mathcal O((n+d)/\log\delta)}\cdot\operatorname{poly}(|I|)$.
\end{restatable}

For parameters satisfying this stronger density condition, the corollary
also gives an algorithm for finding collisions in the hash-function
family of~\cite{DBLP:conf/focs/SotirakiZZ18}.

Equal Vector Sums replaces the integers in Equal Sums by integer
vectors: the two subsets must have the same sum in every coordinate.
Such equalities also appear in bounds on the number of nonzero variables
needed in integer programs~\cite{ES06}.
We assume that every input vector is nonzero, since a zero vector and
the empty subset already give equal sums.

\begin{prob}[Equal Vector Sums]{def:integer-sums}
\item[Input:] Nonzero vectors $\mathbf{a}_1,\ldots,\mathbf{a}_m\in\mathbb Z^n$
with $\delta>1$ for the parameters below.
\item[Output:] Distinct subsets $S,T\subseteq[m]$ with
$\sum_{i\in S}\mathbf{a}_i=\sum_{i\in T}\mathbf{a}_i$ over the integers.
\item[Parameters:] $n$, $m$, $\Delta\coloneqq\max_i\|\mathbf{a}_i\|_\infty$,
$q\coloneqq2^{\ceil{\log(m\Delta+1)}}$, and $\delta\coloneqq m/(n\log q)$.
\end{prob}

For $n=1$ and positive inputs, this is a dense version of Equal Sums.
The following corollary extends the dense scalar algorithm
of~\cite{BhatiDJRW26} to several coordinates.

\begin{restatable}[Equal Vector Sums]{corollary}{equalVectorSums}\label{cor:integer-sums}
Equal Vector Sums (Problem~\ref{def:integer-sums}) can be solved in time
$q^{\mathcal O(n/\log\delta)}\cdot\operatorname{poly}(|I|)$ for $\delta > 1$.
\end{restatable}
Common elements can be deleted from the two subsets,
so, if $A$ has the input vectors as its columns, we obtain a
nonzero $\mathbf{z}\in\{-1,0,1\}^m$ with $A\mathbf{z}=0$ over the
integers.

Minkowski's problem asks for a short nonzero \emph{lattice vector}
in the \emph{modular lattice} $\Lambda_q^\perp(A)$ prescribed by a~matrix $A$
(we recall basic facts about lattices in \Cref{sec:lattice-prelim}).

\begin{prob}[Dense Minkowski]{def:minkowski}
\item[Input:] A matrix $A\in\mathbb Z_q^{n\times m}$
with $m>n\log q$.
\item[Output:] A nonzero vector $\mathbf{z}\in\Lambda_q^\perp(A)$ with
$\|\mathbf{z}\|_\infty\leq\det(\Lambda_q^\perp(A))^{1/m}$.
\item[Parameters:] $n$, $m$, $q$, and $\delta\coloneqq m/(n\log q)$.
\end{prob}

The density promise makes the bound on the infinity norm less than $2$, so the
required vector has coordinates in $\{-1,0,1\}$.  Minkowski's theorem is an
important tool in mathematics that guarantees a vector satisfying the bound
above~\cite{Minkowski10}.  Sotiraki et
al.~\cite{DBLP:conf/focs/SotirakiZZ18} showed that its general computational
form belongs to $\mathsf{PPP}$. Here we give a subexponential-time algorithm for it in
the dense regime.

\begin{restatable}[Dense Minkowski]{corollary}{denseMinkowski}\label{cor:minkowski}
Dense Minkowski (Problem~\ref{def:minkowski}) can be solved in time
$q^{\mathcal O(n/\log\delta)}\cdot\operatorname{poly}(|I|)$ for $\delta>1$.
\end{restatable}
Blichfeldt's theorem is a basic tool in the geometry of numbers and
gives a standard proof of Minkowski's theorem~\cite{Blichfeldt14}.
Its general computational form is $\mathsf{PPP}$-complete.
Several lattice approximation problems can be solved using a
polynomial number of calls to a Blichfeldt
solver~\cite{DBLP:conf/focs/SotirakiZZ18}.
We consider the following version for integer points in a box and
modular lattices, defined in \Cref{sec:lattice-prelim}.

\begin{prob}[Dense Blichfeldt]{def:weak-blichfeldt}
\item[Input:] A matrix $A\in\mathbb Z_q^{n\times m}$,
$\mathbf{p}\in\mathbb Z^m$, and powers of two $L_1,\ldots,L_m\geq2$.
Let
$B=\{\mathbf{x}\in\mathbb Z^m:p_i\leq x_i<p_i+L_i\text{ for all }i\in[m]\}$
and require $|B|>q^n$.
\item[Output:] Distinct points $\mathbf{x},\mathbf{y}\in B$ with
$\mathbf{x}-\mathbf{y}\in\Lambda_q^\perp(A)$.
\item[Parameters:] $n$, $q$, and $\delta\coloneqq\log|B|/(n\log q)$.
\end{prob}

\begin{restatable}[Dense Blichfeldt]{corollary}{weakBlichfeldtBoxes}\label{cor:weak-blichfeldt}
Dense Blichfeldt (Problem~\ref{def:weak-blichfeldt}) can be solved in time
$q^{\mathcal O(n/\log\delta)}\cdot\operatorname{poly}(|I|)$ for $\delta>1$.
\end{restatable}
\Cref{fig:pigeonhole-classes} shows the reductions and their corollary numbers.

\subsection{Technical overview}

We adapt the block-merging algorithm of~\cite{BhatiDJRW26} to prove
\Cref{thm:SISalg}. For $\delta = \Oh(1)$, brute force gives the required bound.
For $\delta = \omega(1)$, we choose the largest power of two $k\geq2$ with
$k\leq m$ and $k/\log k\leq\delta/2$, and put $W \coloneqq {2^{\ceil{\delta n \log q / (2k)}}}$.
Divide the columns into $k$ blocks $T_0^0,\ldots,T_{k-1}^0$, each of size at least
$\lfloor m/k\rfloor$. Each block supplies at least $W$ binary vectors, as follows.
For each $i\in[0,k)$, we construct a set $Y_i^0$ of size $W$ by taking any
$W$ distinct binary vectors supported on $T_i^0$ and computing $A\bx$;
this is possible since there are $2^{|T_i^0|} \geq 2^{\lfloor m/k\rfloor} \geq W$ such vectors.
A repeated image within a block is already a collision.
For every $\mathbf y\in Y_i^0$, we store the binary vector
$f(\mathbf y)$ that produced $\mathbf y$.
All blocks $T^0_i$ are pairwise disjoint, so different pairs have different merged binary vectors.

Let $\boldsymbol{\mu}^0=\mathbf1_n$ be the vector of $n$ ones.
To obtain $\boldsymbol{\mu}^j$, start with $\boldsymbol{\mu}^{j-1}$ and repeat
the following step $\log W$ times: multiply a coordinate with the smallest
current value by $2$, breaking ties by the smallest index.
Throughout this, the entries of $A$ are integers in
$[0,q)$, and all products and sums are computed over the integers.
See~\Cref{alg:collision}.

\begin{algorithm}[H]
\SetAlgoNoLine
\caption{Algorithm behind \Cref{thm:SISalg}.}
\label{alg:collision}
Partition $[m]$ into $T_0^0\uplus\cdots\uplus T_{k-1}^0$ as above\;
$Y_i^0\gets W$ images from $T_i^0$ for every $i\in[0,k)$,
and return a collision if found\;
\For{$j=1,\ldots,\log k$ \textnormal{\textbf{and}} $i=0,\ldots,k/2^j-1$}{
    Form all $W^2$ pairs from $Y_{2i}^{j-1}$ and $Y_{2i+1}^{j-1}$\;
    Return a collision if two pair sums agree, and let $Z$ be the sums\;
    Partition $Z$ into classes
    $\mathcal B_{\mathbf{p}}\gets\{\mathbf{z}\in Z:\mathbf{z}\bmod\boldsymbol{\mu}^j=\mathbf{p}\}$\;
    $\mathbf{p}^\star\gets\arg\max_{\mathbf{p}} |\mathcal B_{\mathbf{p}}|$\;
    $Y_i^j\gets$ any $W$ entries of $\mathcal B_{\mathbf{p}^\star}$\;
}
Pick any two elements of $Y_0^{\log k}$ and report a~collision between them
\end{algorithm}

Throughout the procedure, we keep the following invariants: (i) all vectors of $Y^j_i$ are pairwise congruent modulo $\boldsymbol{\mu}^{j}$ and (ii) each $Y_{i}^j$ has size $W$. For $j=0$, both conditions are true. For the induction step, the main point is that if we did not find a collision, then $|Z| = |Y_{2i}^{j-1}| |Y_{2i+1}^{j-1}| = W^2$. The algorithm partitions $Z$ into at most $W$ disjoint classes, so at least one of them will have size at least $W$.

We claim that this procedure eventually finds a correct collision.
Collisions found within the \textbf{for}-loop are obviously correct, so assume
no such collision within this loop was found. By the invariants, the set
$Y_0^{\log{k}}$ contains at least two elements. These are congruent modulo
$\boldsymbol{\mu}^{\log k}$. Since $W^{\log k} \geq q^n$, this implies that
all moduli in $\boldsymbol{\mu}^{\log k}$ are (multiples of) $q$.
This in turn means that these two elements of $Y_0^{\log{k}}$ map to the same value $\bmod\;q$,
which means they form a~collision that can be reported in the final step of the algorithm!

For the running time, observe that each merge takes
$W^2\cdot\operatorname{poly}(|I|)$ time.  By the maximality of $k$, we have
$W^2\leq4q^{4n/\log\delta}$. Since there are fewer than $m$ merges, this yields
the running time in \Cref{thm:SISalg}.  See~\Cref{sec:main-proof} for the full
proof.

\paragraph*{Weak constrained SIS in the dense setting.}
The reduction for \Cref{cor:dense-wcsis} must find a collision whose
two vectors also have the prescribed image $\mathbf{b}$ under a binary
invertible matrix $G\in\mathbb Z_q^{d\times m}$.
Informally, the first $d\log q$ columns of $G$ form $d$ blocks of $\log q$
columns: row $i$ has the binary weights $1,2,4,\ldots,2^{\log q-1}$ in
block $i$, zeros in earlier blocks, and arbitrary entries in all later
columns (see \Cref{def:binary-invertible}).
For example, for $q=16$ and $d=3$, every matrix of the following form
is binary invertible:
\[
    G=\left[\begin{array}{cccc|cccc|cccc|c}
        2^0&2^1&2^2&2^3&*&*&*&*&*&*&*&*&*\\
        0&0&0&0&2^0&2^1&2^2&2^3&*&*&*&*&*\\
        0&0&0&0&0&0&0&0&2^0&2^1&2^2&2^3&*
    \end{array}\right].
\]
Each star is an arbitrary element of $\mathbb Z_{16}$.
The three diagonal blocks encode four binary bits each, and the last
column is $G_1$.
Split $A=[A_0\ A_1]$ and $G=[G_0\ G_1]$, with the first
$d\log q$ columns in $A_0$ and $G_0$.

Apply \Cref{thm:SISalg} to the matrix
$C=[A_1^{\mathsf T},G_1^{\mathsf T}]^{\mathsf T}$.
The resulting binary vectors $\bz\neq\bz'$ have equal images under both
$A_1$ and $G_1$.
This equality under $G_1$ is essential: both vectors now require the
same residual image $\mathbf{b}-G_1\bz$ under $G_0$.

Binary invertibility gives a unique binary vector $\bu$ with
$G_0\bu=\mathbf{b}-G_1\bz$.
We compute it from the last row to the first, setting each gadget block
to the binary decomposition of its residual right-hand side.
The same $\bu$ completes both vectors, giving
$\bx=[\bu^{\mathsf T},\bz^{\mathsf T}]^{\mathsf T}$ and
$\by=[\bu^{\mathsf T},(\bz')^{\mathsf T}]^{\mathsf T}$.
They satisfy $G\bx=G\by=\mathbf{b}$, and $\bu$ cancels in
$A\bx-A\by=A_1(\bz-\bz')=0$.
The instance has $n+d$ rows and $m-d\log q$ columns, giving
exactly the density and running time in \Cref{cor:dense-wcsis}.
Constructing the instance and completing the vectors take polynomial
time.

The other reductions follow the same template: construct a dense
SIS instance and convert its collision to the required solution. For
integer vector sums, the modulus is large enough to turn a modular
equality into an integer equality. For Dense Blichfeldt, binary expansion
encodes the box, and Dense Minkowski uses the box $\{0,1\}^m$.
\Cref{sec:corollary-proofs} gives the full reductions.

\Needspace{10\baselineskip}
\subsection{Related work}

Sotiraki, Zampetakis, and Zirdelis~\cite{DBLP:conf/focs/SotirakiZZ18}
proved that \emph{constrained SIS} is $\mathsf{PPP}$-complete and weak
constrained SIS is $\mathsf{PWPP}$-complete. Their construction also
gives the hash family $\mathcal H_{\mathsf{cSIS}}$, whose collision
problem is weak constrained SIS. These completeness results depend on
the additional constraints in $G$.

For Equal Subset Sum without the pigeonhole promise, successive
improvements~\cite{MNPW19,RW26} yield the best known running time of
$1.7067^m\cdot\operatorname{poly}(|I|)$ for $m$ inputs.

Under the pigeonhole promise, faster algorithms are known for Equal Sums.
Jin and Wu~\cite{JW24} gave a
randomized $2^{0.4m}\cdot\operatorname{poly}(|I|)$-time algorithm.
Jin, Williams, and Zhang~\cite{JWZ25} improved it to
$2^{m/3}\cdot\operatorname{poly}(|I|)$.
Bhati et al.~\cite{BhatiDJRW26} obtain
$M^{o(1)}\cdot\operatorname{poly}(m)$ time for integers bounded by $M$
when $m=\omega(\log M)$.
We extend their block-merging method to modular systems and use their
scalar reduction to transfer the conditional lower bound.

A related line of work uses Frobenius-number bounds to obtain faster
algorithms for total Unbounded Subset Sum.
For distinct positive integers $a_1,\ldots,a_m$ with greatest common divisor
$1$, this problem asks for nonnegative integer coefficients $x_i$ such that
$\sum_i a_i x_i=b$.
Writing $M=\max_i a_i$, Aggarwal, Joux, Santha, and
W\k{e}grzycki~\cite{AJSW24} give a polynomial-time algorithm when
$b\geq\varepsilon M^2$ for any fixed $\varepsilon>0$, and analogous
results for integer programming with equalities.
Chen, Mao, and Zhang~\cite{CMZ25} subsequently gave an
expected $\mathcal O(m\log M)$-time algorithm when
$b\geq C M^2/m$ for a sufficiently large absolute constant $C$,
using constructive results from additive combinatorics.

For the modular pigeonhole problem in one row, Allcock et al.~\cite{AHJ22}
give a $\widetilde{\mathcal O}(2^{m/2})$-time algorithm
whenever $q<2^m$. Using only $\log q+1$ columns gives the
$q^{1/2}\cdot\operatorname{poly}(|I|)$ bound for one-row SIS.
A classical algorithm of Regev, described by Chen, Liu, and
Zhandry~\cite[Appendix~A]{CLZ22}, solves binary collision SIS in
polynomial time when $m\geq(n+1)^{\log q}$.
\Cref{thm:SISalg} gives a subexponential bound already when
$m=\omega(n\log q)$.

The number-balancing reductions of Hoberg et al.~\cite{HRRY17}
concern general convex bodies.
\Cref{cor:minkowski} concerns only the modular lattices
$\Lambda_q^\perp(A)$ with $m>n\log q$.

Wagner's generalized birthday algorithm~\cite{Wagner02} merges lists
while imposing stronger conditions on their partial sums at each level.
Ducas, Engelberts, and Loyer~\cite{DEL25} use a lattice form of this
method and discrete Gaussian sampling to give subexponential algorithms
for SIS in the infinity norm. Their guarantees concern random matrices,
polynomial prime moduli, and norm bounds larger than one.
Our algorithm uses a power-of-two modulus and works for every input
matrix. It keeps a largest congruence class at each merge, whose size
is bounded below by counting.

%% file: figure1.tex
\begin{figure}[!b]
\centering
\begin{minipage}[c]{0.49\textwidth}
\centering
\parbox[c][55mm][c]{\linewidth}{\centering
\begin{tikzpicture}[x=0.85cm,y=0.88cm,font=\sffamily\footnotesize,
    every node/.style={inner sep=2pt,align=center},
    reduction/.style={->,>=stealth,line width=0.8pt,
        shorten <=1pt,shorten >=1pt}]
\node[font=\sffamily\bfseries\footnotesize] (ppp) at (3.7,4.35) {PPP};
\node (blichfeldt) at (1.2,3.75) {BLICHFELDT};
\node (csis) at (5.45,3.75) {cSIS};
\node (minkowski) at (1.2,2.25) {MINKOWSKI};
\node[font=\sffamily\bfseries\footnotesize] (pwpp) at (3.7,2.25) {PWPP};
\node (wcsis) at (6.15,2.25) {weak-cSIS};
\node (equal) at (1.2,0.5) {EQUAL-SUMS};
\node (sis) at (3.7,0.5) {SIS};
\draw[reduction] ([yshift=2pt]blichfeldt.east) -- ([yshift=2pt]ppp.west);
\draw[reduction] ([yshift=-2pt]ppp.west) -- ([yshift=-2pt]blichfeldt.east);
\draw[reduction] ([yshift=2pt]ppp.east) -- ([yshift=2pt]csis.west);
\draw[reduction] ([yshift=-2pt]csis.west) -- ([yshift=-2pt]ppp.east);
\draw[reduction] (pwpp.north) -- (ppp.south);
\draw[reduction] ([yshift=2pt]pwpp.east) -- ([yshift=2pt]wcsis.west);
\draw[reduction] ([yshift=-2pt]wcsis.west) -- ([yshift=-2pt]pwpp.east);
\draw[reduction] (minkowski.north) -- (blichfeldt.south);
\draw[reduction] (equal.north) -- (minkowski.south);
\draw[reduction] (sis.north west) -- (minkowski.south east);
\draw[reduction] (sis.north) -- (pwpp.south);
\end{tikzpicture}}
\end{minipage}\hfill
\begin{minipage}[c]{0.49\textwidth}
\centering
\parbox[c][55mm][c]{\linewidth}{\centering
\begin{tikzpicture}[x=0.84cm,y=0.82cm,font=\sffamily\footnotesize,
    every node/.style={inner sep=2pt,align=center},
    reduction/.style={->,>=stealth,line width=0.8pt,
        shorten <=1pt,shorten >=1pt},
    edge label/.style={font=\sffamily\footnotesize,fill=white,inner sep=1.5pt},
    middle/.style={minimum height=7.5mm},
    bottom/.style={minimum height=5.5mm}]
\node[font=\sffamily\bfseries\footnotesize] at (3.6,5.3) {Dense regime};
\node (sis) at (3.6,4.4) {BINARY SIS};
\node[middle] (wcsis) at (0.95,2.9) {weak-cSIS};
\node[middle] (integer) at (6.25,2.9) {EQUAL VECTOR\\SUMS};
\node[bottom] (weakblichfeldt) at (3.6,1.7) {DENSE\\BLICHFELDT};
\node[bottom] (minkowski) at (3.6,-0.5) {DENSE\\MINKOWSKI};
\node[bottom] (pess) at (6.25,-0.3) {EQUAL-SUMS};
\draw[reduction] (wcsis.north) -- node[edge label,above,sloped]
    {Cor.~\ref{cor:dense-wcsis}} (sis.south west);
\draw[reduction] (integer.north) -- node[edge label,above,sloped]
    {Cor.~\ref{cor:integer-sums}} (sis.south east);
\draw[reduction] (weakblichfeldt.north) -- node[edge label,left,pos=0.45]
    {Cor.~\ref{cor:weak-blichfeldt}} (sis.south);
\draw[reduction] (minkowski.north) -- node[edge label,left]
    {Cor.~\ref{cor:minkowski}} (weakblichfeldt.south);
\draw[reduction] (pess.north) -- node[edge label,right] {$n=1$} (integer.south);
\end{tikzpicture}}
\end{minipage}
\caption{Selected reductions from~\cite[Figure~3]{DBLP:conf/focs/SotirakiZZ18}
(left). Roadmap of our reductions (right).}
\label{fig:pigeonhole-classes}
\end{figure}
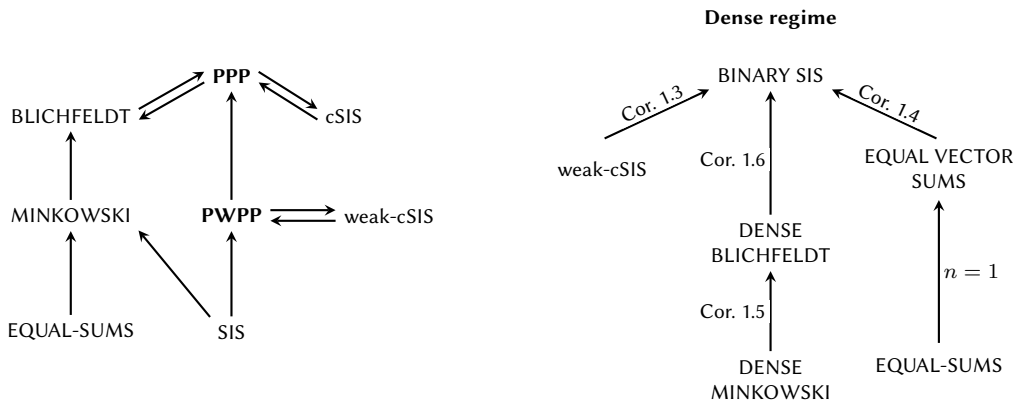

%% file: prelim.tex
\section{Preliminaries}\label{sec:prelim}

We write $[m]=\{1,\ldots,m\}$, $[a,b] = \{a, \ldots, b\}$, $[a, b) = [a, b-1]$.
All logarithms have base $2$, and $|I|$ denotes the binary input length.
Vectors are column vectors written in bold (e.g., $\mathbf{x}$), with binary
vectors additionally underlined (e.g., $\bx$).
We write $[\mathbf{u}^{\mathsf T},\mathbf{v}^{\mathsf T}]^{\mathsf T}$ for their vertical concatenation.
Matrix equations over $\mathbb{Z}_q$ are read modulo $q$.
Vector congruences and remainders are taken coordinatewise.
The vectors $\mathbf{0}_n$
and $\mathbf{1}_n$ have all coordinates zero and one, respectively.
For $B,C\subseteq\mathbb{Z}_q^n$, their sumset is
$B+C=\{\mathbf{b}+\mathbf{c}:\mathbf{b}\in B,\ \mathbf{c}\in C\}$.

\subsection{Total search and pigeonhole classes}\label{sec:pigeonhole-classes}

A search problem is \emph{total} if every input has a solution.
The class $\mathsf{TFNP}$ consists of total search problems whose
solutions have size polynomial in input size and that can be verified in polynomial
time. A Boolean circuit is a directed acyclic graph with ordered input and
output nodes and gates computing $\neg$, $\land$, $\lor$, or $\xor$.
A negation gate has one input and the other gates have two. Evaluating
its gates in topological order defines a function
$\mathcal C:\{0,1\}^r\to\{0,1\}^s$,
where $r$ and $s$ count the input and output nodes, respectively.
Its size $|\mathcal C|$ is the number of nodes.

The search problem \textsc{Pigeonhole Circuit} takes a Boolean circuit
$\mathcal C:\{0,1\}^r\to\{0,1\}^r$ and asks for either an input
$\bx$ with $\mathcal C(\bx)=\mathbf{0}_r$ or distinct inputs
$\bx,\by$ with $\mathcal C(\bx)=\mathcal C(\by)$.
If no input maps to zero, the $2^r$ inputs have at most $2^r-1$
possible outputs, so a collision exists.
The class $\mathsf{PPP}$ consists of search problems reducible to
\textsc{Pigeonhole Circuit} by polynomial-time search
reductions~\cite{Papadimitriou}.

The search problem \textsc{Collision} takes as input a Boolean circuit
$\mathcal C$ with $r$ inputs and $s<r$ outputs and asks for two distinct
vectors $\bx,\by\in\{0,1\}^r$ such that
$\mathcal C(\bx)=\mathcal C(\by)$.
There are at least twice as many inputs as possible outputs.
The class $\mathsf{PWPP}$ consists of search problems reducible to
\textsc{Collision} by polynomial-time search
reductions~\cite{DBLP:conf/focs/SotirakiZZ18}.
Thus, $\mathsf{PWPP}\subseteq\mathsf{PPP}\subseteq\mathsf{TFNP}$.
A problem is complete for one of these classes if it belongs to the
class and every problem in the class reduces to it.

Observe that binary collision SIS (\Cref{def:sis}) belongs to $\mathsf{PWPP}$. Namely, one can 
encode $A\bx$ in binary to get a circuit with $m$ input bits and
$n\log q<m$ output bits. The reductions in \Cref{sec:corollary-proofs} also place Equal Vector
Sums, Dense Minkowski, and Dense Blichfeldt in $\mathsf{PWPP}$ under
their respective input promises.

\emph{Equal Sums}~\cite{Papadimitriou} takes positive integers $a_1,\ldots,a_m$ with
$\sum_i a_i<2^m-1$ and asks for distinct subsets $S,T\subseteq[m]$
such that $\sum_{i\in S}a_i=\sum_{i\in T}a_i$.
Write $M=\max_i a_i$.
There are $2^m$ subsets and fewer than $2^m$ possible sums, so a
solution exists.

\subsection{Lattices and their search problems}\label{sec:lattice-prelim}

A \emph{full-rank lattice} in $\mathbb{R}^m$ is a set
$\mathcal L=D\mathbb{Z}^m$, where $D\in\mathbb{R}^{m\times m}$ is
nonsingular. Its columns form a \emph{basis} of $\mathcal L$, and its
\emph{determinant} is $\det(\mathcal L)=|\det D|$.
For computational problems, bases have rational entries and are
encoded in binary.
We use the infinity norm $\|\mathbf{v}\|_\infty=\max_i|v_i|$ and the
Euclidean norm $\|\mathbf{v}\|_2=(\sum_i v_i^2)^{1/2}$.

A \emph{lattice coset} is a translate
$\mathbf{p}+\mathcal L=\{\mathbf{p}+\mathbf{v}:\mathbf{v}\in\mathcal L\}$.
Two points $\mathbf{x},\mathbf{y}$ lie in the same coset exactly when
$\mathbf{x}-\mathbf{y}\in\mathcal L$.
If $\mathcal L\subseteq\mathbb{Z}^m$ is a full-rank integer lattice,
then $\det(\mathcal L)$ equals the number of its cosets in
$\mathbb{Z}^m$.

For $A\in\mathbb{Z}_q^{n\times m}$, we use the \emph{modular lattice}
$ \Lambda_q^\perp(A)=
    \{\mathbf{z}\in\mathbb{Z}^m:A\mathbf{z}=0\pmod q\}$.
It contains $q\mathbb{Z}^m$ and therefore has full rank.
Two integer vectors lie in the same coset of $\Lambda_q^\perp(A)$
exactly when they have the same image under $A$.
Since there are at most $q^n$ images, we have
$1\leq\det(\Lambda_q^\perp(A))\leq q^n$.

Minkowski's theorem~\cite{Minkowski10} guarantees that every full-rank
lattice $\mathcal L\subseteq\mathbb{R}^m$ contains a nonzero vector
$\mathbf{v}$ with
$\|\mathbf{v}\|_\infty\leq\det(\mathcal L)^{1/m}$.
Minkowski's problem asks for such a vector given a lattice basis.
Notice that \Cref{def:minkowski} restricts the input to modular lattices with
$m>n\log q$.

Blichfeldt's theorem~\cite{Blichfeldt14} states that a measurable set
$S\subseteq\mathbb{R}^m$ of volume greater than $\det(\mathcal L)$
contains distinct points $\mathbf{x},\mathbf{y}$ with
$\mathbf{x}-\mathbf{y}\in\mathcal L$.
For integer points, the counting argument is enough: any set
$S\subseteq\mathbb{Z}^m$ with $|S|>\det(\mathcal L)$ contains two
points in the same coset of a full-rank integer lattice $\mathcal L$.
This guarantees a solution to Dense Blichfeldt (\Cref{def:weak-blichfeldt}).
The more general computational formulations of Minkowski and
Blichfeldt are studied in~\cite{DBLP:conf/focs/SotirakiZZ18}.

\Needspace{9\baselineskip}
\subsection{Binary invertibility and weak constrained SIS}
\label{sec:binary-invertibility}

For a power-of-two $q\geq2$, the gadget vector is the column
$\boldsymbol{\gamma}_q=(2^0,2^1,\ldots,2^{\log q-1})^\top$.
For $a\in\mathbb{Z}_q$, let $\operatorname{bin}(a)\in\{0,1\}^{\log q}$
be the binary decomposition of its representative in $\{0,\ldots,q-1\}$,
with the least significant bit first. Thus,
$\boldsymbol{\gamma}_q^\top\operatorname{bin}(a)=a\pmod q$.
For vectors, $\operatorname{bin}$ concatenates the decompositions of
successive coordinates.

\begin{definition}[Binary invertible matrix~\cite{DBLP:conf/focs/SotirakiZZ18}]\label{def:binary-invertible}
Let $q\geq2$ be a power of two and $m\geq d\log q$.
A matrix $G\in\mathbb{Z}_q^{d\times m}$ is \emph{binary invertible} if
$G=[I_d\otimes\boldsymbol{\gamma}_q^\top+U\ \ V]$, where
$U\in\mathbb{Z}_q^{d\times d\log q}$ satisfies $U_{i,j}=0$ for $j\leq i\log q$,
and $V\in\mathbb{Z}_q^{d\times(m-d\log q)}$ is arbitrary.
Here $\otimes$ is the Kronecker product, and matrix indices start at $1$.
\end{definition}

In words, divide the first $d\log q$ columns into $d$ blocks.
Row $i$ contains the binary weights\linebreak  $1,2,4,\ldots,2^{\log q-1}$ in block
$i$, zeros in earlier blocks, and arbitrary entries in later columns.
For every $\bz\in\{0,1\}^{m-d\log q}$ and target
$\mathbf{b}\in\mathbb{Z}_q^d$, there is a unique
$\bu\in\{0,1\}^{d\log q}$ such that
$G[\bu^{\mathsf T},\bz^{\mathsf T}]^{\mathsf T}=\mathbf{b}$.

To compute $\bu$, process rows $i=d,d-1,\ldots,1$.
Subtract the contribution of the already chosen coordinates from
$b_i$ modulo $q$, and write this remainder in binary to fill block $i$.
This is the only choice satisfying row $i$, and it does not affect the
rows already processed. Thus, the procedure gives the unique $\bu$.
After computing $V\bz$, it uses
$\mathcal O(d^2\log q)$ modular arithmetic operations.

A standard brute-force algorithm gives the following bound.

\begin{observation}[Brute-force]\label{lem:wcsis-enumeration}
Every weak constrained SIS instance, also with an arbitrary target
$\mathbf{b}\in\mathbb{Z}_q^d$, can be solved in time
$q^n\cdot\operatorname{poly}(|I|)$.
\end{observation}
\begin{proof}
Since $2^{m-d\log q}>q^n$, choose $q^n+1$ distinct binary vectors
of length $m-d\log q$. Complete each vector to a binary vector $\bx$
with $G\bx=\mathbf{b}$ using the procedure above.
The completed vectors are distinct and their images under $A$ lie in
$\mathbb Z_q^n$, so two images coincide. Completing the vectors,
evaluating $A$, and sorting the images to find a repetition take
$q^n\cdot\operatorname{poly}(|I|)$ total time.
\end{proof}

\subsection{The lattice hypothesis and lower bound}
\label{sec:lattice-lower-bound}

Bhati et al.~\cite[Section~5]{BhatiDJRW26} give a conditional lower bound for
Equal Sums.
For a full-rank lattice $\mathcal L\subseteq\mathbb{R}^N$,
$\lambda_N(\mathcal L)$ is the least radius of a Euclidean ball centered
at $0$ containing $N$ linearly independent lattice vectors.
The hypothesis concerns the \emph{$N$-th Successive Minima Problem},
which asks for one nonzero vector of length at most a given multiple
of $\lambda_N(\mathcal L)$.
It states that, for every constant $c\geq1$, there is an
$\varepsilon>0$ such that no
$\mathcal O(2^{\varepsilon N})$-time algorithm, given $\mathcal L$,
returns with high probability a nonzero $\mathbf{v}\in\mathcal L$ with
$\|\mathbf{v}\|_2\leq\mathcal O(N^c)\lambda_N(\mathcal L)$.
Here and below, the running times include a polynomial factor in the
binary encoding length of the lattice basis.

\SISlowerbound*

\begin{proof}
Recall that Bhati et al.~\cite[Section~5]{BhatiDJRW26} rule out a
$M^{o(1/\log m)}\cdot\operatorname{poly}(m)$-time algorithm for
Equal Sums. Their reduction gives hard instances
with $m$ positive integers bounded by
$M=2^{\Theta(\sqrt m\log m)}$.
Choosing the smallest power of two $q>mM$ gives $mM<q\leq2mM$.
Form the binary collision SIS instance
$A=[a_1,\ldots,a_m]\in\mathbb Z_q^{1\times m}$.
All subset sums lie in $[0,q)$, so a modular collision $A\bx=A\by$
gives distinct subsets $S=\{i:x_i=1\}$ and $T=\{i:y_i=1\}$ with equal
integer sums. The instance has $n=1$,
$\log q=\Theta(\log M)$, and
$\delta=m/\log q=\Theta(\sqrt m/\log m)=m^{1/2+o(1)}$.
In particular, $\delta>1$ for sufficiently large $m$, so this is a valid
binary collision SIS instance.
The reduction takes polynomial time, so a binary collision SIS algorithm
with running time $q^{o(1/\log\delta)}\cdot\operatorname{poly}(|I|)$
would solve Equal Sums in time
$M^{o(1/\log m)}\cdot\operatorname{poly}(m)$.
\end{proof}

%% file: main_proof.tex
\section{Main proof}\label{sec:main-proof}


\input{proof_sis}

%% file: proof_sis.tex
\SISalg*
\begin{proof}
    If $\delta \leq 4$, we can find such distinct $\bx$, $\by$ by brute force in time $q^{\Oh(n)} \cdot \operatorname{poly}(|I|)$.
    Henceforth assume $\delta > 4$.
    Let $k \geq 2$ be the largest integer power of 2 such that $k/\log{k} \leq \delta/2$ and $k \leq m$.

    
    For $j \in [0, \log{k}]$ define a~family $T^j = \{T_i^j : i \in [0, \tfrac{k}{2^j})\}$ of integer subintervals of $[m]$:
    \[T_i^j \coloneqq \left[ \floor{\tfrac{i \cdot 2^j \cdot m}{k}} + 1,\ \floor{\tfrac{(i + 1) \cdot 2^j \cdot m}{k}} \right] \qquad\text{for }i \in [0, \tfrac{k}{2^j}), \]
    and the families of Boolean indicator vectors of the subsets of $T_i^j$:
    \[X_i^j \coloneqq \set{\bx \in \set{0, 1}^m : \forall_{ \ell \in [m] \setminus T_i^j} \ \bx_\ell = 0}.\]
    Note that for $j \in [\log{k}]$ we have
    \[T_i^{j} = T_{2i}^{j - 1} \uplus T_{2i + 1}^{j - 1} \quad\quad \text{and} \quad\quad X_i^j = X_{2i}^{j-1} + X_{2i+1}^{j-1} = \set{\bx^1 + \bx^2 \mid \bx^1 \in X_{2i}^{j - 1}, \bx^2 \in X_{2i + 1}^{j - 1}}.\]
    Moreover, each $T_i^j$ has size at least $2^j\cdot \floor{\tfrac{m}{k}}$.
    Finally, let $W := {2^{\ceil{\delta n \log q / (2k)}}}$. 
    \begin{claim}
        $W \leq 2^{\floor{m / k}}$.
    \end{claim}
    \begin{proof}
        We have $\log{W} = \ceil{\delta n \log q / (2k)} = \ceil{m / (2k)} \leq \floor{m / k}$, where the inequality holds since $m / k \geq 1$.
    \end{proof}
    In particular, each set $X_i^0$ has size at least $2^{\floor{m / k}} \geq W$.

    
    \medskip

    We now present the algorithm.
    In the entire routine, we regard the input (the entries of $A$) as integers in $[0, q)$ instead of remainders $\bmod\; q$, and all products and sums in this proof are computed over the integers.

    We begin by initializing the integer vector $\boldmu^0=\mathbf{1}_n$.
    Next, for each $i \in [0, k)$ we build a~set $Y_i^0$ of size $W$ by taking any $W$ vectors $\bx \in X_i^0$ and adding $A\bx$ to $Y_i^0$.
    If, during this step, we come across two vectors $\bx^1, \bx^2$ such that $A\bx^1 = A\bx^2$, we declare success immediately with $(\bx^1, \bx^2)$ as the sought collision.
    Otherwise, create an~explicit mapping from the entries of $Y_i^0$ to the binary vectors of $X_i^0$ and call it $f$.
    That is, for every $\mathbf{y}\in Y_i^0$ write $f(\mathbf{y})$ to refer to the vector $\bx \in X_i^0$ with $A\bx=\mathbf{y}$ that produced $\mathbf{y}$.
    We stress that here $f$ is understood separately for each set $Y_i^0$.
    
    Now the algorithm proceeds in $\log{k}$ \emph{rounds}, indexed by $j = 1, \ldots, \log k$.
    On the $j$th round:
    \begin{enumerate}
        \item Construct a vector $\boldmu^j$ from $\boldmu^{j-1}$ by performing ${\log{W}}$ coordinate doublings, every time selecting a coordinate with the smallest value, breaking ties by the smallest index.
        \item For each $i \in [0, \tfrac{k}{2^j})$:
        \begin{itemize}
            \item Compute ${Z} \coloneqq \set{\mathbf{y}^1 + \mathbf{y}^2 : \mathbf{y}^1 \in Y_{2i}^{j-1}, \mathbf{y}^2 \in Y_{2i + 1}^{j-1}}$. If we spot a collision, i.e., we find two different pairs of vectors $(\mathbf{y}^1, \mathbf{y}^2), ((\mathbf{y}^1)', (\mathbf{y}^2)')$ with $\mathbf{y}^1+\mathbf{y}^2=(\mathbf{y}^1)'+(\mathbf{y}^2)'$, we immediately terminate and return the collision $f(\mathbf{y}^1)+f(\mathbf{y}^2)$, $f((\mathbf{y}^1)')+f((\mathbf{y}^2)')$. 
            Otherwise, for each $\mathbf{z}=\mathbf{y}^1+\mathbf{y}^2$, we define the mapping $f(\mathbf{z})=f(\mathbf{y}^1)+f(\mathbf{y}^2)$.

            


            \item Create the partition $\mathcal{P}$ of $Z$ into subsets based on the congruence classes of the vectors in $Z$ $\bmod\; \boldmu^j$; that is, for each $\mathbf{y}, \mathbf{y}'$ in the same part, we have $\mathbf{y} \equiv \mathbf{y}' \pmod{\boldmu^j}$.


            \item Choose any set $P \in \mathcal{P}$ of largest size and let $Y_i^j \subseteq P$ be an arbitrary set of $W$ vectors.
        \end{itemize}
    \end{enumerate}

    \subparagraph{Correctness.}
    It is straightforward to see that throughout the algorithm the following invariants are preserved for every $j \in [0, \log{k}]$ and $i \in [0, \tfrac{k}{2^j})$:
    \begin{itemize}
        \item each $Y_i^j$ has size $W$;
        \item for each $\mathbf{y}, \mathbf{y}' \in Y_i^j$, we have $\mathbf{y} \equiv \mathbf{y}' \pmod{\boldmu^j}$;
        \item for each $\mathbf{y}\in Y_i^j$, we have $A\,f(\mathbf{y})=\mathbf{y}$ and $f(\mathbf{y})\in X_i^j$;
        \item $\prod_{\ell\in[n]} \boldmu_\ell^j =  2^{\log{W}\cdot j}$; each $\boldmu^j_\ell$ is a power of $2$, and for every $\ell, \ell'$, we have $\boldmu_\ell^j\leq 2\boldmu_{\ell'}^{j}$.
    \end{itemize}

    Now consider a~round $j \in [\log k]$ and an~index $i \in [0, \tfrac{k}{2^j})$.
    If a~collision is spotted, the output is correct by the invariants above; so assume otherwise.
    First, the size of the partition $\mathcal{P}$ is at most $W$: indeed, by the invariant in round $j-1$, all vectors of $Y_{2i}^{j-1}$ are pairwise congruent to each other $\bmod\; \boldmu^{j-1}$, and the same holds for $Y_{2i+1}^{j-1}$.
    Hence, all vectors of $Z$ are also pairwise congruent to each other $\bmod\; \boldmu^{j-1}$.
    $\mathcal{P}$ partitions these vectors into congruence classes $\bmod\; \boldmu^j$; hence there are at most
    \[
        \left( \prod_{\ell\in [n]} \boldmu_\ell^j \right) /
        \left( \prod_{\ell\in [n]} \boldmu_\ell^{j-1} \right) = 2^{\log W} = W
    \]
    congruence classes, and so $|\mathcal{P}| \leq W$.
    Then, since $|Z| = |Y_{2i}^{j-1}| \cdot |Y_{2i+1}^{j-1}| = W^2$, we conclude that some part $P \in \mathcal{P}$ has size at least $W$, as required.

    Suppose the algorithm finished without reporting a~collision, and consider the set $Y_0^{\log k}$.
    All vectors therein are pairwise congruent $\bmod\; \boldmu^{\log k}$.
    We have
    \[
        \prod_{\ell \in [n]} \boldmu_\ell^{\log k} =
        2^{\log k \cdot \log W} \geq
        2^{\log k \cdot \delta n \log q / (2k)} \geq
        2^{n \log q} = q^n,
    \]
    where the last inequality follows from the definition of $k$ such that $k / \log k \leq \delta/2$.
    Thus, all coordinates of $\boldmu^{\log k}$ are at least $q$: since $q$ is a~power of two, any coordinate smaller than $q$ would need to be upper-bounded by $q / 2$.
    But then the invariant would guarantee that all the remaining coordinates would be upper-bounded by $q$.
    Then the product of all entries of $\boldmu^{\log k}$ would not exceed $q^n / 2$ -- a~contradiction.
    Therefore, all coordinates of $\boldmu^{\log k}$ are actually divisible by $q$.
    Now take any two distinct vectors $\mathbf{y}, \mathbf{y}' \in Y_0^{\log k}$.
    We have $\mathbf{y} \equiv \mathbf{y}' \pmod{\boldmu^{\log k}}$, and so $\mathbf{y} \equiv \mathbf{y}' \pmod{q \cdot \mathbf{1}_n}$.
    Then the corresponding binary vectors $\bx = f(\mathbf{y})$, $\bx' = f(\mathbf{y}')$ yield a~collision pair $\pmod{q}$, which we can return from the algorithm.

    \subparagraph{Time analysis.}
    On initialization, we compute all sets $Y_i^0$ for $i \in [0, k)$ in total time $\Oh(k \cdot W\cdot \tfrac{m}{k}\cdot n) = W \cdot \operatorname{poly}(|I|)$: we have $|Y_i^0| = W$, and each element of $Y_i^0$ is formed by multiplying a~vector of $X_i^0$ -- with at most $\left\lceil \tfrac{m}{k} \right\rceil$ non-zero entries -- by $A$. 
    Now we analyze the time spent on the $j$th round. Fix an~iteration $i \in [0, \frac{k}{2^j})$. Each set $Y_i^j$ has size $W$. The set $Z$ has size $\Oh(W^2)$, so computing $Z$ and partitioning it take $W^2 \cdot \operatorname{poly}(|I|)$ time. Finding the largest $P \in \mathcal{P}$ requires $\Oh(W)$ time.
    Therefore, the $j$th round requires $W^2 \cdot \operatorname{poly}(|I|)$ time, so the total running time of the algorithm is 
    \[ W^2 \cdot \operatorname{poly}(|I|) = 2^{2\ceil{\delta n \log q /(2k)}} \cdot \operatorname{poly}(|I|) \leq q^{\delta n / k} \cdot \operatorname{poly}(|I|).\]
    The last polynomial factor absorbs the constant from rounding.
    As we have chosen $k \geq 2$ to be the largest power of $2$ such that $k \leq m$ and $2k \leq \delta \log{k}$, one of the following must hold:

    \begin{itemize}
        \item $2k > m$. Then $q^{\delta n / k} = 2^{\delta n \log q / k} = 2^{m / k} \leq 4$, so
        $q^{\delta n / k} \cdot \operatorname{poly}(|I|) = \operatorname{poly}(|I|)$.
        \item $4k > \delta \log(2k)$. Then $\delta n / k < 4n / \log(2k) < 4n / \log \delta$, with the second inequality following from $4k > \delta \log(2k) \geq 2\delta$.
        Hence, $q^{\delta n / k} \cdot \operatorname{poly}(|I|) = q^{\Oh(n / \log \delta)} \cdot \operatorname{poly}(|I|)$.
    \end{itemize}

    The required time bounds follow.
\end{proof}

%% file: corollary_proofs.tex
\section{Reductions and proofs of the corollaries}\label{sec:corollary-proofs}

Each proof in this section constructs an instance of another problem in polynomial
time and converts every solution back to the original instance.
Dense Minkowski reduces to Dense Blichfeldt, which reduces to binary
collision SIS. The other two reductions go directly to binary collision
SIS. \Cref{fig:pigeonhole-classes} shows these reductions.

\denseWcSIS*

\begin{proof}
We prove the bound also for an arbitrary target $\mathbf{b}\in\mathbb Z_q^d$.
Let $n' = n+d$ and $m'=m-d\log q$ and split the input matrices as
$A=[A_0\ A_1]$ and $G=[G_0\ G_1]$, where the first blocks have
$d\log q$ columns. Construct the binary collision SIS instance
\[
    C = \left[
        \begin{matrix} A_1 \\ G_1 \end{matrix}
        \right] \in \mathbb{Z}_q^{n' \times m'}.
\]
The density is $m'/(n'\log q)>1$, so the instance is valid.

Let $\bz\neq\bz'$ be any solution of this SIS instance.
Then $A_1\bz=A_1\bz'$ and $G_1\bz=G_1\bz'$ modulo $q$.
The completion procedure following \Cref{def:binary-invertible} computes
the unique binary prefix $\bu$ satisfying $G_0\bu=\mathbf{b}-G_1\bz$.
Since $G_1\bz=G_1\bz'$, the same prefix completes both suffixes to $\mathbf{b}$.

Return $\bx=[\bu^{\mathsf T},\bz^{\mathsf T}]^{\mathsf T}$ and
$\by=[\bu^{\mathsf T},(\bz')^{\mathsf T}]^{\mathsf T}$.
They are distinct because their suffixes differ. The equal images under
$G_1$ imply $G\bx=G\by=\mathbf{b}$, and
\[
    A\bx-A\by=A_1(\bz-\bz')=0\pmod q.
\]
Constructing $C$ and completing the prefix take polynomial time in
the input length. So \Cref{thm:SISalg} with $n'$ rows and $m'$ columns gives
$q^{\mathcal O(n'/\log\delta)}\cdot\operatorname{poly}(|I|)$ time.
\end{proof}

\equalVectorSums*

\begin{proof}
Given the vectors $\mathbf{a}_1,\ldots,\mathbf{a}_m$, compute $\Delta$ and the smallest
power of two $q>m\Delta$. Construct $A\in\mathbb Z_q^{n\times m}$
by reducing each column $\mathbf{a}_i$ modulo $q$.
The promise $\delta=m/(n\log q)>1$ makes the instance valid.

Let $\bx\neq\by$ be any solution. Define
$S=\{i:x_i=1\}$ and $T=\{i:y_i=1\}$, which are distinct subsets.
Their difference of integer sums is
\[
    \mathbf{v}=\sum_{i\in S}\mathbf{a}_i-\sum_{i\in T}\mathbf{a}_i
      =\sum_{i=1}^m(x_i-y_i)\mathbf{a}_i.
\]
We have $\mathbf{v}=0\pmod q$ and
$\|\mathbf{v}\|_\infty\leq\sum_i|x_i-y_i|\|\mathbf{a}_i\|_\infty
\leq m\Delta<q$.
Each coordinate of $\mathbf{v}$ is a multiple of $q$ of absolute value less than
$q$, so $\mathbf{v}=0$ over the integers. Common elements may be removed
from $S$ and $T$ if disjoint subsets are desired.

Since $q\leq2m\Delta$, we have $\log q=\mathcal O(\log m+\log\Delta)$ and so the bit length of the reduced instance,
$O(nm \log q)$, is polynomial in $|I|$.
Constructing $A$ and decoding $S,T$ take polynomial time in the input length, so \Cref{thm:SISalg}
gives the stated $q^{\mathcal O(n/\log\delta)}\cdot\operatorname{poly}(|I|)$
bound.
\end{proof}

\weakBlichfeldtBoxes*

\begin{proof}

Given a Dense Blichfeldt instance (Problem~\ref{def:weak-blichfeldt}), write
$A=[\mathbf{a}_1\ \cdots\ \mathbf{a}_m]$ and let $C\in\mathbb Z_q^{n\times\log|B|}$
be the matrix with columns $\mathbf{c}_{i,j}=2^j \mathbf{a}_i\pmod q$ for
$i\in[m]$ and $0\leq j<\log L_i$.
Since $\log|B|>n\log q$, this is a valid binary collision SIS instance with
density $\log|B|/(n\log q)=\delta$.
For a binary vector $\bu$ with these coordinates, define
$\phi(\bu)\in\mathbb Z^m$ by
\[
    \phi(\bu)_i\coloneqq \sum_{j=0}^{\log L_i-1}2^j u_{i,j}.
\]
Binary expansion gives a bijection from $\{0,1\}^{\log|B|}$ to
$\prod_i\{0,\ldots,L_i-1\}=B-\mathbf{p}$. We have
$C\bu=A\phi(\bu)\pmod q$.
Given a solution $\bu\neq\bu'$ with $C\bu=C\bu'$,
return $\mathbf{x}=\mathbf{p}+\phi(\bu)$ and $\mathbf{y}=\mathbf{p}+\phi(\bu')$.
The bijection gives distinct points in $B$, and
\[
    A(\mathbf{x}-\mathbf{y})=A(\phi(\bu)-\phi(\bu'))=C(\bu-\bu')=0\pmod q.
\]
Hence $\mathbf{x}-\mathbf{y}\in\Lambda_q^\perp(A)$, as required.

For the running-time analysis, the constructed matrix has $n\log|B|$ entries of $\log q$ bits, and its columns
can be computed by successive doublings modulo $q$.
Also, decoding the vectors and adding $\mathbf{p}$ take polynomial time.
So if we apply \Cref{thm:SISalg} with $n$ rows and $\log|B|$ columns,
we obtain an algorithm running in time $q^{\mathcal O(n/\log\delta)}\cdot\operatorname{poly}(|I|)$.
\end{proof}

\denseMinkowski*

\begin{proof}
    
We reduce Dense Minkowski (Problem~\ref{def:minkowski}) to
Dense Blichfeldt (Problem~\ref{def:weak-blichfeldt}).
Keep $A$ and $q$, set $\mathbf{p}=0$, and set $L_i=2$ for every $i\in[m]$, so 
the box is $B=\{0,1\}^m$.  The constructed Dense Blichfeldt instance is valid because
$\log|B|/(n\log q) > 1$ since $m > n \log{q}$.

Given a solution $\mathbf{x}\neq \mathbf{y}$ for this box, return $\mathbf{z}=\mathbf{x}-\mathbf{y}$.
Then $\mathbf{z}\in\Lambda_q^\perp(A)$ is nonzero and
$\mathbf{z}\in\{-1,0,1\}^m$, so $\|\mathbf{z}\|_\infty=1$.
Since $\Lambda_q^\perp(A)$ is a full-rank sublattice of $\mathbb Z^m$,
its determinant is a positive integer, hence at least $1$.
Thus $\|\mathbf{z}\|_\infty=1\leq\det(\Lambda_q^\perp(A))^{1/m}$,
which is Minkowski's required bound. Every nonzero vector of this
integer lattice has infinity norm at least $1$, so $\mathbf{z}$ is also shortest.

The instance can be constructed in polynomial time.
Applying \Cref{cor:weak-blichfeldt} with the same $n$, $q$, and $\delta$
yields the required $q^{\mathcal O(n/\log\delta)}\cdot\operatorname{poly}(|I|)$ time.
\end{proof}